\batchmode
\documentclass[12pt]{article}    
\usepackage{hyperref}
\usepackage{graphics}
\usepackage{epsfig}
\usepackage{color}
\usepackage[10pt]{moresize}
\usepackage[round,authoryear]{natbib}
\usepackage{diagbox}
\usepackage{mathrsfs}
\begin{document}
\newtheorem{definition}{Definition}
\newtheorem{theorem}{Theorem}
\newtheorem{example}{Example}
\newtheorem{corollary}{Corollary}
\newtheorem{lemma}{Lemma}
\newtheorem{proposition}{Proposition}
\newenvironment{proof}{{\bf Proof:\ \ }}{\qed}
\newcommand{\qed}{\rule{0.5em}{1.5ex}}
\newcommand{\bfg}[1]{\mbox{\boldmath $#1$\unboldmath}}

\newcommand{\fraca}[2]{\displaystyle\frac{#1}{#2}}

\def \B{{\rm I\kern -2.1pt B\hskip 1pt}}
\def \N{{\rm I\kern -2.1pt N\hskip 1pt}}
\def \R{{\rm I\kern -2.2pt  R\hskip 1pt}}
\begin{center}

\thispagestyle{empty}

\section*{An Alternative Representation of the Negative Binomial--Lindley Distribution. New Results and Applications}

\vskip 0.2in {\sc \bf Emilio G\'omez--D\'eniz$^a$ and Enrique
Calder\'in--Ojeda$^b$

{\small\it $^a$Department of Quantitative Methods in Economics and
T$i$DES Institute.
University of Las Palmas de Gran Canaria, Spain.}\\

\noindent{\small\it$^b$  Centre for Actuarial Studies, Department
of Economics, The University of Melbourne, Australia}}
\end{center}

\begin{abstract}
\noindent In this paper we present an alternative representation
of the Negative Binomial--Lindley distribution recently proposed
by Zamani and Ismail (2010) which shows some advantages over the
latter model. This new formulation provides a tractable model with
attractive properties which makes it suitable for application not
only in insurance settings but also in other fields where
overdispersion is observed. Basic properties of the new
distribution are studied. A recurrence for the probabilities of
the new distribution and an integral equation for the probability
density function of the compound version, when the claim
severities are absolutely continuous, are derived. Estimation
methods are discussed and a numerical application is given.
\end{abstract}

\noindent {\bf Keywords}: Lindley Distribution, Mixture, Negative Binomial Distribution, EM algorithm, Insurance.

\vspace{1cm}


\vfill
\subsubsection*{Acknowledgements}
Authors thank Ministerio de Econom\'ia y
Competitividad (project ECO2013-47092).

\vspace{1cm}

\noindent Address for correspondence: Emilio G\'omez D\'eniz,
Department of Quantitative Methods, University of Las Palmas de
Gran Canaria, 35017--Las Palmas de Gran Canaria, Spain. E--mail:
{\ttfamily emilio.gomez-deniz@ulpgc.es}

\section{Introduction}
\noindent Distribution mixtures define one of the most important ways to
obtain new probability distributions in applied probability and
operational research. In this sense, and looking for a more
flexible alternative to the Poisson distribution, especially under
the overdispersion phenomena (variance larger than the mean), the negative binomial obtained as a mixture of
Poisson and gamma distributions. In a similar fashion negative binomial--Pareto (Klugman et al. (2008)) and
Poisson--inverse Gaussian distribution, also known as Sichel distribution (Willmot (1987)) have been
proposed in actuarial contexts, particularly in the automobile
insurance setting and other fields where the empirical data seems to show contagious or heterogeneity.

Recently, Zamani and Ismail (2010) proposed a mixture of the negative binomial distribution with parameters $r>0$ and $0<p<1$. For this purpose, they allow the parameter $p=1-\exp(-\lambda)$, $\lambda>0$ to follow a Lindley distribution. The resulting mixture model has been also applied also recently in the context of accident analysis by Lord and Geedipally (2011). 

In this paper we present an alternative representation of this mixture model which can be written in terms of the confluent hypergeometric function and the Pochhammer symbol. This representation shows some advantages over the one previously introduced in the literature. The new formulation provides a tractable model with attractive properties which makes
it suitable for application not only in insurance settings but
also in other fields where overdispersion is observed. Additionally, some basic
properties of the new distribution that were not examined in Zamani and Ismail (2010) are introduced. Some of these features include the unimodality and overdispersion among other properties. A recurrence for the probabilities of the new distribution together with an integral equation
for the probability density function of the compound version, when
the claim severities are absolutely continuous, are also presented. Estimation methods are discussed by factorial moment and maximum likelihood methods. In addition to these methods, an EM type algorithm is introduced when the triple mixture Poisson--Gamma--Lindley is considered.

The remaining of the paper proceeds as follows. In Section
\ref{s1} we introduce the basic distributions assumed, the negative binomial and Lindley distributions. Section \ref{s2} analyzes the basic properties of the model
including the probability function, factorial and
ordinary moments, recurrence, overdispersion and unimodality. Some methods of estimations are given in Section \ref{s3}. Section \ref{s4} studies the compound negative binomial--Lindley
distribution. An integral equation is derived for the probability density function
of the compound version, when the claim severities are
absolutely continuous, from the basic principles assumed in the collective risk model. Applications are provided in Section \ref{s5} and the work finishes with the conclusions.

\section{Basic distributions}\label{s1}
In this section we introduce the definition and some basic
properties of the negative binomial and Lindley distributions. A classical negative binomial distribution
with probability mass function
\begin{eqnarray}
p_{r,\lambda}(x)={r+x-1\choose x}\left(\frac{1}{1+\lambda}\right)^r \left(\frac{\lambda}{1+\lambda}\right)^x,\quad x=0,1,\dots\label{nbd}
\end{eqnarray}
will be denoted as $X\sim \mathscr{NB}(r,\lambda)$, where $r>0$ and
$\lambda>0$. As they will be needed later, we remind some
characteristics of this distribution. The mean, variance and the factorial moment, $\mu_{[k]}(X) = E[X(X-1)\cdots
(X-k+1)]$, of a negative binomial
distribution (see Balakrishnan and Nevzorov (2003)) are
respectively given by,
\begin{eqnarray}
E(X) &=& r\lambda,\nonumber\\
var(X) &=& r\lambda(1-\lambda),\nonumber\\
\mu_{[k]}(X) &=& (r)_k \lambda^k,\;\;k=1,2,\dots,\label{fm}
\end{eqnarray}
where $(a)_n=\Gamma(a+n)/\Gamma(a)$ represents the Pochhammer symbol and $\Gamma(s)=\int_0^\infty \tau^{s-1}e^{-\tau}\,d\tau$
denotes the complete gamma function.

The probability generating function of a random variable $X$ following the probability function (\ref{nbd}) is given by
\begin{eqnarray*}
G_X(z)=(1+\lambda(1-z))^{-r},\quad |z|\leq 1.
\end{eqnarray*}

Henceforward, we will use $X\sim \mathscr{NB}(r,\lambda)$ to denote a random variable $X$ that follows a negative binomial distribution with parameters $r>0$ and $\lambda>0$.

Although the continuous one--parameter Lindley distribution, initially introduced by Lindley (1958), has
not been widely used in the past, however it has become a popular probabilistic model in the last decade, in part as a result of its simplicity and excellent performance in practice.  In this regard, it has been chosen as mixing distribution when the parameter of the
Poisson distribution is considered random (Sankaran (1971)).
In that paper it is shown that the resulting distribution provided a better fit to the
empirical set of data considered than the negative binomial and
Hermite distributions. Recently, a good deal of attention has been
given to this probability density function and new papers
have been included in the statistical literature.  See Ghitany et al. (2008) and references therein, Ghitany et al. (2013), G\'omez--D\'eniz et al. (2013); among others. A random variable $\Lambda$ has a Lindley distribution if its probability density function is
given by,
\begin{equation}
g_{\theta}(\lambda)=\frac{\theta^2}{1+\theta}(1+\lambda)\exp(-\theta \lambda),\;\lambda>0\label{ld}
\end{equation}
where $\theta>0$. In the following, we will denote as $\Lambda\sim \mathscr{L}(\theta)$ for a random variable that follows a Lindley distribution.

\section{Representation of the negative binomial--Lindley distribution}\label{s2}
In this paper we introduce  an alternative representation of the Negative Binomial--Lindley distribution recently proposed by Zamani and Ismail (2010) that has several advantages over the latter model. This new formulation
provides a tractable model with attractive properties that makes
it suitable for applications not only in insurance settings but
also in other fields where the overdispersion phenomenon is observed.

\begin{definition}We say that a random variable $X$ has a negative
binomial--Lindley distribution if it admits the stochastic
representation:
\begin{eqnarray}
X|\lambda &\sim& \mathscr{NB}(r,\lambda),\label{def11}\\
\lambda & \sim & \mathscr{L}(\theta),\label{def12}
\end{eqnarray}
with $r,\lambda,\theta>0$. We will denote this distribution by $X\sim
\mathscr{NBL}(r,\theta)$.
\end{definition}

The next result provides closed--form expressions for the probability mass
function and factorial moments.
\begin{theorem} Let $X\sim \mathscr{NBL}(r,\theta)$ be a negative
binomial--Lindley distribution defined in
(\ref{def11})-(\ref{def12}). Some basic properties are:
\begin{itemize}
\item[(a)] The probability mass function is given by
\begin{eqnarray}
p_{r,\theta}(x)=\frac{\theta^2 (r)_x}{1+\theta}\,\mathscr{U}(x+1,3-r,\theta),\quad x=0,1,\dots,\label{nr}
\end{eqnarray}
where
\begin{eqnarray*}
\mathscr{U}(a,b,z)=\frac{1}{\Gamma(a)}\int_0^{\infty} \tau^{a-1}(1+\tau)^{b-a-1}\exp(-z\tau)\,d\tau,
\end{eqnarray*}
is the confluent hypergeometric function (see Gradshteyn and Ryzhik (1994), p. 1085, formula 9211–-4).

\item[(b)] The factorial
moment of order $k$ is given by
\begin{eqnarray}
\mu_{[k]}(X)&=& \frac{(r)_k k!(k+\theta+1)}{(1+\theta)\theta^k},\label{fmnbl}
\end{eqnarray}
with $k=1,2,\dots$

\item[(c)] The mean and the variance are given by,
\begin{eqnarray*}
E(X) &=& \frac{2+\theta}{1+\theta}\frac{r}{\theta},\\
var (X) &=& \frac{r (6 (1+\theta)+(4+\theta) ((1+\theta) \theta +r \theta)+2 r)}{\theta ^2 (1+\theta)^2}.
\end{eqnarray*}
\end{itemize}
\end{theorem}
\begin{proof}
\begin{itemize}
\item[(a)]
The probability function of $X$ can be obtained by using
the well--known compound formula,
\begin{eqnarray*}
p_{r,\theta}(x)=\int_{0}^{\infty}p_{r,\lambda}(x)g(\lambda)\,d\lambda,
\end{eqnarray*}
and rearranging parameters. Here, $g(\lambda)$ is the probability density function of the Lindley distribution in (\ref{ld}).
\item[(b)]
The factorial moments of order $k$ are obtained making use of (\ref{fm}) and having into account that

$$E(\mu_{[k]}(X) )=E_{\lambda}(E(\mu_{[k]}(X|\lambda) ))=(r)_k  E_{\lambda}(\lambda^k).$$

\item[(c)]Finally, the mean and variance are straightforwardly derived from (\ref{fmnbl}).

\end{itemize}
\end{proof}

Now, by using the fact that
\begin{eqnarray*}
\mathscr{U}(a,b,z)=z^{1-b}\,\mathscr{U}(a-b+1,2-b,z)
\end{eqnarray*}
and for computational purposes it is convenient rewrite (\ref{nr}) as
\begin{eqnarray}
p_{r,\theta}(x)=\frac{\theta^r (r)_x}{1+\theta}\,\mathscr{U}(x+r-1,r-1,\theta).\label{nr1}
\end{eqnarray}

Observe that the special case $r=1$ provides the geometric--Lindley distribution.

\begin{theorem}The probability function of an $\mathscr{NBL}$
distribution can be evaluated by the recursive formula
\begin{equation}\label{recu1}
p_{r,\theta}(x)=\frac{r+x-1}{x}p_{r,\theta}(x-1)-\frac{r}{x}p_{r+1,\theta}(x-1),\;\;x=1,2,\dots
\end{equation}
where $p_r(0)=\theta^r \exp(\theta)\Gamma(2-r,\theta)/(1+\theta)$ and $\Gamma(a,z)=\int_z^{\infty}\tau^{a-1}\exp(-\tau)\,d\tau$ is the incomplete gamma function.
\end{theorem}
\begin{proof}
For the negative binomial distribution with pmf
\begin{eqnarray*}
p_{r,\lambda}(x)={r+x-1\choose x}\left(\frac{1}{1+\lambda}\right)^r \left(\frac{\lambda}{1+\lambda}\right)^x,\quad x=0,1,\dots
\end{eqnarray*}
we have the simple recursion
\begin{equation}\label{recu2}
p_{r,\lambda}(x)=\frac{\lambda}{1+\lambda}\frac{r+x-1}{x}p_{r,\lambda}(x-1),\;\;x=1,2,\dots
\end{equation}
Using the definition of a $\mathscr{NBL}$ distribution and
(\ref{recu2}) we get
\begin{eqnarray*}
  p_{r,\theta}(x) &=& \int_0^\infty p_{r,\lambda}(x)g(\lambda)\,d\lambda \\
   &=& \frac{r+x-1}{x}\int_0^\infty \frac{\lambda}{1+\lambda}p_{r,\lambda}(x-1)g(\lambda)\,d\lambda \\
   &=& \frac{r+x-1}{x}\int_0^{\infty}\left(1-\frac{1}{1+\lambda}\right)p_{r,\lambda}(x-1)g(\lambda)\,d\lambda\\
   &=& \frac{r+x-1}{x}\left[p_{r,\theta}(x-1)-\int_{0}^{\infty}\frac{1}{1+\lambda}p_{r,\lambda}(x-1)g(\lambda)\,d\lambda\right] .
\end{eqnarray*}
Now, since
\begin{eqnarray*}
\int_0^\infty  \frac{1}{1+\lambda}p_{r,\lambda}(x-1)g(\lambda)\,d\lambda &=& \frac{r}{r+x-1}\int_0^\infty p_{r+1,\lambda}(x-1)g(\lambda)\,d\lambda\\
&=& \frac{r}{r+x-1}p_{r+1,\theta}(x-1),
\end{eqnarray*}
we obtain (\ref{recu1}).
\end{proof}

Calculation of the probabilities are now easy and they do not require
the use of the confluent hypergeometric function. This result can be also obtained by using expression (12) in Willmot (1993).

\begin{proposition}
Let $Z$ a positive and continuous random variate following a gamma distribution with probability density function $f(z)\propto \lambda^{-r}\exp(-z /\lambda)$, $r>0$, $\lambda>0$, and assume that $\lambda$ is random following a Lindley distribution with parameter $\theta>0$. Then, the unconditional probability density function of $Z$ results,
\begin{eqnarray}
f(z)= \frac{2 \theta^{r/2+1} z^{r/2-1}}{(1+\theta)\Gamma(r)}\left(z K_{r-2}(2\sqrt{\theta z})+
\sqrt{\theta z}K_{r-1}(2\sqrt{\theta z})\right),\quad z>0,\label{uncond}
\end{eqnarray}

\noindent where $K_n(\cdot)$ is the modified Bessel function of the second kind.
\end{proposition}
\begin{proof}
The result follows by computing the integral
\begin{eqnarray*}
f(z) &=& \frac{\theta^2   z^{r-1}}{(1+\theta)\Gamma(r)}\int_0^{\infty}\lambda^{-r}(1+\lambda)\exp\left(-z/\lambda-\lambda\theta)\right)\,d\lambda\\
&=&\frac{\theta^2   z^{r-1}}{(1+\theta)\Gamma(r)}\left(\int_0^{\infty}\lambda^{-r}\exp\left(-z/\lambda-\lambda\theta)\right)\,d\lambda\right.\\
&+&\left.\int_0^{\infty}\lambda^{1-r}\exp\left(-z/\lambda-\lambda\theta)\right)\,d\lambda\right).
\end{eqnarray*}

Hence the result.
\end{proof}

\begin{proposition} The probability density function in (\ref{uncond}) is log--concave for $r\geq 1$ and therefore unimodal.
\end{proposition}
\begin{proof} It is well--known that the gamma distribution defined above is log--concave for $r\geq 1$. Now the result follows by using Prekopa's Theorem (see Lynch (1999)) and having into account that the probability density function (\ref{uncond}) is obtained as a mixture of a gamma distribution with the Lindley distribution which is also log--concave.
\end{proof}

As a consequence of the last Proposition, we have the following result.
\begin{proposition}
The discrete distribution with probability function given in (\ref{nr}) is unimodal for $r\geq 1$.
\end{proposition}
\begin{proof}
It is direct consequence of a result provided in Holgate (1970).
\end{proof}

Next result shows that the $\mathscr{NBL}$ discrete distribution is a Poisson mixture distribution.
\begin{proposition} The discrete distribution with probability function given in (\ref{nr}) is a Poisson $(\mathscr{P}(\sigma)$, $\sigma>0)$ mixture distribution with mixing distribution given in (\ref{uncond}).
\end{proposition}
\begin{proof} Using $Z\sim \mathscr{GL}(r,\theta)$ to denote a random variate which follows the probability density function (\ref{uncond}) and $\Sigma\sim \mathscr{G}(r,\lambda)$ when $\Sigma$ follows a gamma probability density function, it is obvious that the mixture
$\mathscr{NB}(r,\lambda)\bigwedge_{\lambda}\mathscr{L}(\theta)$
can be written as
\begin{eqnarray}
p_{r,\theta}(x) &=& \mathscr{NB}(r,\lambda)\bigwedge_{\lambda}
\mathscr{L}(\theta)=\left(\mathscr{P}(\sigma)\bigwedge_{\sigma}\mathscr{G}(r,\lambda)\right)\bigwedge_{\lambda}
\mathscr{L}(\theta)\nonumber\\
&=&
\mathscr{P}(\sigma)\bigwedge_{\sigma}\left(\mathscr{G}(r,\lambda)
\bigwedge_{\lambda}\mathscr{L}(r,\theta)\right).\label{mixt}\end{eqnarray}

Hence the proposition.
\end{proof}

In the following, we state two more results (without proof) addressing the calculation of the posterior expectations and overdispersion of the $\mathscr{NBL}$ distribution.

By using Proposition 10 in Karlis and Xekalaki (2005) the posterior
expectation of $\lambda^r$ given $x$ can be computed as follows
\begin{eqnarray*}
E(\lambda^s|x)=\frac{\Gamma(x+s)\,p_{r,\theta}(x+s)}{\Gamma(x+1)\,p_{r,\theta}(x)},
\end{eqnarray*}
for $s$ taking positive or negative values.

Furthermore,since the $\mathscr{NBL}$ distribution distribution arises from a mixture of
a Poisson distribution, the variance--to--mean ratio is greater
than one (see Karlis and Xekalaki (2005) and Sundt and Vernic (2009), p.66) which implies that the new distribution is
overdispersed (variance larger than mean).

\section{Estimation of parameters}\label{s3}
Let $\tilde x=(x_1,\dots,x_n)$ be a random sample from model
(\ref{nr}). A simple polynomial equation can be
obtained by equating the first two sample and  theoretical factorial moments derived from (\ref{fmnbl}). Let $\tilde f_1=m_{[1]}(X)$ and $\tilde f_2=m_{[2]}(X)$ the sample version of the
factorial moments. Then, we have the following system of equations,
\begin{eqnarray}
\tilde f_1 &=& \mu_{[1]}(X),\label{eq1}\\
\frac{\tilde f_2}{\tilde f_1} &=& \frac{m_{[2]}(X)}{m_{[1]}(X)}.\label{eq2}
\end{eqnarray}

After some computations we have the expression
\begin{eqnarray*}
\theta(2+\theta)^2\tilde f_2-2\tilde f_1(3+\theta)\left[\theta(1+\tilde f_1(1+\theta))+2\right]=0,
\end{eqnarray*}
that depends solely on the parameter $\theta$ and it can be solved numerically. Finally, by plugging this estimated parameter into (\ref{eq1}), the estimate of the parameter $r$ is obtained.

These moment estimates can be used as starting values in the calculation of the maximum likelihood estimates. The maximum likelihood estimates can be obtained directly by maximizing the log--likelihood function, which is straightforwardly derived from (\ref{nr1}), is given by
\begin{eqnarray}
\ell(\tilde x;r,\theta) &=& n\left[r\log(\theta)-\log(1+\theta)-\log(\Gamma(r))\right]\nonumber\\
&&+\sum_{i=1}^n \left[
\log(\Gamma(r+x_i))+\log(\mathscr{U}(x_i+r-1,r-1,\theta))\right].\label{logl}
\end{eqnarray}

 Since the global maximum of the log-likelihood surface
is not guaranteed, different initial values of the parametric
space can be considered as a seed point. In this sense, by using
the {\tt FindMaximum} function of Mathematica software package
v.11.0 (Wolfram (2003)) and comparing by using other different methods such as Newton, PrincipalAxis and
QuasiNewton (all of them available in that package) the same
result is obtained. Finally, the standard errors of the parameter
estimates have been approximated by inverting the Hessian matrix.
These also can be obtained by approximating the Hessian matrix and
recovering it from the Cholesky factors.
\subsection{Estimation by EM type algorithm}
Maximum likelihood estimates can also be achieved by means of the EM algorithm to avoid to use the confluent hypergeometric function when maximizing the log--likelighood function (\ref{logl}). The algorithm could be implemented by using the fact that the $\mathscr{NBL}$ discrete distribution arises as a mixture of the Poisson distribution where the Poisson parameter $\sigma$ follows the distribution given by (\ref{uncond}). However as the latter probabilistic family is not a member of the exponential family of probability distributions, the conditional expectations require in the Expectation E--step do not coincide with their sufficient statistics. For that reason, to put into action this algorithm we make use of the mixture representation given in (\ref{mixt}). Given the observations $\tilde{x}$ and the missing observations $\tilde{\sigma}=(\sigma_1,\dots,\sigma_n)^{\top}$ and $\tilde{\lambda}=(\lambda_1,\dots,\lambda_n)^{\top}$, the complete probability mass function is

\begin{eqnarray*}
f(\tilde{x},\tilde{\sigma},\tilde{\lambda}|r,\theta)&=&\prod_{i=1}^{n}  f(x_i,\sigma_i,\lambda_i|r,\theta)\\
&=&\prod_{i=1}^{n}  f(x_i|\sigma_i)\times f(\sigma_i|\lambda_i,r)\times f(\lambda_i|\theta)
\end{eqnarray*}
and the complete likelihood is

\begin{equation}
\ell(r,\theta;\tilde{x},\tilde{\sigma},\tilde{\lambda})\propto \sum_{i=1}^{n}\log f(\sigma_i|\lambda_i,r)+ \sum_{i=1}^{n}\log f(\lambda_i|\theta).\label{compli}
\end{equation}

In the E--step, the expectation of (\ref{compli}), conditional of the observations $\tilde{x}$ and given the parameter estimates $r$ and $\theta$ is given by

\begin{eqnarray}
E(\ell(r,\theta;\tilde{x},\tilde{\sigma},\tilde{\lambda})|\tilde{x},{\hat r},{\hat \theta})&\propto& E\left( \sum_{i=1}^{n}\log f(\sigma_i|\lambda_i,r)|\tilde{x},{\hat r},{\hat \theta}\right)\nonumber\\
&+&E\left( \sum_{i=1}^{n}\log f(\lambda_i|\theta)|\tilde{x},{\hat r},{\hat \theta}\right)\label{condiexp}.
\end{eqnarray}
\noindent where ${\hat r}$ and ${\hat \theta}$ denotes estimates of parameter $r$ and $\theta$ respectively and $E\equiv E_{\tilde{\lambda}|\tilde{x},\tilde{\sigma},\tilde{\theta}}$ with $\tilde{\theta}=(r,\theta)$.

In the M--step, the updated parameter estimates are obtained from maximizing the quantity (\ref{condiexp}) with respect to $r$ and $\theta$. Particularly, by conditional independence, we have

\begin{eqnarray}
{\cal A}(\theta)&=&E\left(\sum_{i=1}^{n}\log f(\lambda_i|\theta)|\tilde{x},{\hat r},{\hat \theta}\right)\nonumber\\
&=&\sum_{i=1}^{n}E(\log f(\lambda_i|\theta)|\tilde{x},{\hat r},{\hat \theta})\nonumber\\
&=&2n\log\theta-n\log(1+\theta)+\sum_{i=1}^{n}E(\log(1+\lambda_i)|\tilde{x},{\hat r}, {\hat \theta})+\theta\sum_{i=1}^{n}E(\lambda_i|\tilde{x},{\hat r}, {\hat \theta})\nonumber\\
&\propto& 2n\log\theta-n\log(1+\theta)+\theta\sum_{i=1}^{n}E(\lambda_i|\tilde{x},{\hat r}, {\hat \theta}).\label{E1}
\end{eqnarray}
In a similar fashion we have

\begin{eqnarray}
{\cal B}(r)&=&E\left(E_{\tilde{\sigma}|\tilde{x},\tilde{\theta}}\left(\sum_{i=1}^{n}\log f(\sigma_i|\lambda_i,r)|\tilde{x},{\hat r},{\hat \theta}\right)\right)\nonumber\\
&=&\sum_{i=1}^{n}E\left(E_{\tilde{\sigma}|\tilde{x},\tilde{\theta}}\left(\log f(\sigma_i|\lambda_i,r)|\tilde{x},{\hat r},{\hat \theta}\right)\right)\nonumber\\
&\propto&r\sum_{i=1}^{n}E(\log\lambda_i|\tilde{x},{\hat r}, {\hat \theta})-n\log \Gamma(r)\nonumber\\
&+&(r-1)\sum_{i=1}^{n} E(E_{\tilde{\sigma}|\tilde{x},\tilde{\theta}}(\log\sigma_i|\tilde{x},{\hat r}, {\hat \theta})).\label{E2}
\end{eqnarray}
From these expressions we proceed as follows:
\begin{itemize}
   \item at the E--step the conditional expectation of some functions of $\lambda_i$ are calculated. From the current estimates, $\hat r^{(j)},\;\hat\theta^{(j)}$, we calculate the pseudo--values,

   \begin{eqnarray*}
     r_i&=&\displaystyle E(\lambda_i|x_i,\hat r^{(j)},\hat\theta^{(j)})=\frac{(1+x_i)  \mathscr{U}(2+x_i,4-\hat r^{(j)},\hat\theta^{(j)})}
{\mathscr{U}(1+x_i,3-\hat r^{(j)},\hat\theta^{(j)})},\\
s_i&=&\displaystyle E(\log\lambda_i|x_i,\hat r^{(j)},\hat\theta^{(j)})\\
&=&\frac{\displaystyle\int_{0}^{\infty}\log\lambda_i  {r+x_i-1\choose x_i}\left(\frac{1}{1+\lambda_i}\right)^r \left(\frac{\lambda_i}{1+\lambda_i}\right)^{x_i}   \frac{\theta^2}{1+\theta}(1+\lambda_{i})\exp(-\theta \lambda_{i}) d\lambda_{i}}{\displaystyle\int_{0}^{\infty}{r+x_i-1\choose x_i}\left(\frac{1}{1+\lambda_i}\right)^r \left(\frac{\lambda_i}{1+\lambda_i}\right)^{x_i}   \frac{\theta^2}{1+\theta}(1+\lambda_{i})\exp(-\theta \lambda_{i}) d\lambda_{i} } ,\\
t_i&=&\displaystyle E(\log(\lambda_i+x_i)|x_i,\hat r^{(j)},\hat\theta^{(j)})\\
&=&\frac{\displaystyle\int_{0}^{\infty}\log(\lambda_i+x_i)  {r+x_i-1\choose x_i}\left(\frac{1}{1+\lambda_i}\right)^r \left(\frac{\lambda_i}{1+\lambda_i}\right)^{x_{i}}   \frac{\theta^2}{1+\theta}(1+\lambda_{i})\exp(-\theta \lambda_{i}) d\lambda_{i}}{\displaystyle\int_{0}^{\infty}{r+x_i-1\choose x_i}\left(\frac{1}{1+\lambda_i}\right)^r \left(\frac{\lambda_i}{1+\lambda_i}\right)^{x_{i}}   \frac{\theta^2}{1+\theta}(1+\lambda_{i})\exp(-\theta \lambda_{i}) d\lambda_{i} }.
   \end{eqnarray*}
  \item At the M--step, one maximizes the likelihood of the complete model which reduces to maximization of the mixing distribution. Then, the updated values of the parameters are
      \begin{eqnarray*}
  \hat\theta^{(j+1)}&=&\displaystyle\frac{n-\sum_{i=1}^{n}r_i+\sqrt{(\sum_{i=1}^{n}r_i)^2+6n \sum_{i=1}^{n}r_i+n^2}}{2  \sum_{i=1}^{n}r_i},\\
  \hat r^{(j+1)}&=&\Psi^{-1}\left(\sum_{i=1}^{n}s_i-\sum_{i=1}^{n}t_i+\sum_{i=1}^{n}\Psi(r+x_i)\right),
   \end{eqnarray*}
where $\Psi(\cdot)$ is the digamma function and $\Psi^{-1}(\cdot)$ is the inverse of the digamma function.
   \item If some convergence condition is satisfied then stop iterating, otherwise move back to the E--step for another iteration.
\end{itemize}

\section{Compound model\label{s4}}
Let $X$ be the number of claims in a portfolio of policies in a
time period. Let $Y_i$, $i=1,2,\dots$ be the amount of the $i$-th
claim and $S=\sum_{i=1}^{X}Y_i$ the aggregate claims generated by the portfolio in the period
under consideration. As usual, two fundamental assumptions are
made in risk theory: (1) the random variables $Y_1,Y_2,\dots$ are
independent and identically distributed with cumulative
distribution function $F(y)$ and probability density function $f(y)$ and (2) the random
variables $X,Y_1,Y_2,\dots$ are mutually independent. When an
$\mathscr{NBL}$ is chosen for $X$ (in actuarial setting this is called the primary distribution), the distribution of the aggregate claims $S$ is called compound negative binomial--Lindley
distribution. The cdf of $S$ is:
$$F_S(y)=\sum_{k=0}^\infty F^{*k}(y)\Pr(X=k)$$
where $F^{k\ast}(\cdot)$ denotes the $k$--fold convolution of
$F(\cdot)$ and $\Pr(X=k)$ is given in (\ref{nr}). The main result is given in the next theorem.
\begin{theorem}
If the claim sizes are absolutely continuous random variables with
pdf $f(y)$ for $y>0$, then the pdf $g_s(y;r)$ of the compound
$\mathscr{NBL}$ distribution satisfies the integral equation,
\begin{equation}\label{inteequat}
g_s(y;r)=p_r(0)+\int_{0}^{y} \frac{rs+y-s}{y} g_s(y-s;r)f(s)\,ds
-\int_{0}^{y} \frac{rs}{y} g_s(y-s;r+1)f(s)\,ds.
\end{equation}
\end{theorem}
\begin{proof}
We have that the aggregated claims distribution is given by
\begin{eqnarray*}
g_s(y;r)=\sum_{k=0}^{\infty} p_r(k) f^{k\ast}(y)=p_r(0)
f^{0\ast}(y)+\sum_{k=1}^{\infty} p_r(k) f^{k\ast}(y),
\end{eqnarray*}
where $f^{k\ast}$ denotes the $k$-fold convolution of $f(x)$. Now,
using (\ref{recu1}) we have that
\begin{eqnarray*}
p_r(k)=\left(\frac{r-1}{k}+1\right) p_r(k-1)-\frac{r}{k}
p_{r+1}(k-1),\quad k=1,2,\ldots
\end{eqnarray*}
Then,
\begin{eqnarray*}
\sum_{k=1}^{\infty} p_r(k) f^{k\ast}(y) &=&\sum_{k=1}^{\infty}
\frac{r-1}{k}p_r(k-1) f^{k\ast}(y)+\sum_{k=1}^{\infty} p_r(k-1) f^{k\ast}(y)\\
&-&  \sum_{k=1}^{\infty}\frac{r}{k} p_{r+1}(k-1)f^{k\ast}(y).
\end{eqnarray*}
Now, after some straightforward calculations and using the
identities:
\begin{eqnarray}
f^{k\ast}(y) &=& \int_{0}^{y} f^{(k-1)\ast} (y-s) f(s)\,ds,\quad k=1,2,\ldots \label{convo1}\\
\frac{f^{k\ast}(y)}{k} &=& \int_{0}^{y} \frac{s}{y}
f^{(k-1)\ast}(y-s)f(s)\,ds,\quad k=1,2,\ldots \label{convo2}
\end{eqnarray}
we obtain the result.
\end{proof}

Integral equation (\ref{inteequat}) must be solved numerically.
There are several implementations and algorithms to solve Volterra integral equation of the second kind but, however, they
need to be modified in order to be used in (\ref{inteequat}).
Finally, it is simple to show that if the claim amount
distribution is discrete, expressions (\ref{convo1}) and
(\ref{convo2}) are verified by interchanging $\int_{0}^{y}$ by
$\sum_{s=1}^{y}$ (see Rolski et al. (1999), p.119). Then, the
recursion of compound $\mathscr{NBL}$ distribution is
$$
g_s(y;r)=p_r(0)+\sum_{s=1}^{y} \frac{rs+y-s}{y} g_s(y-s;r)f(s)
-\sum_{s=1}^{y} \frac{rs}{y} g_s(y-s;r+1)f(s).
$$

\section{Numerical application\label{s5}}
In order to test the performance in practice of the $\mathscr{NBL}$ distribution, a simple example dealing where the model introduced in this paper is fitted to an insurance dataset that concerns to the number of automobile liability policies in Zaire (1974) for private cars
(Willmot (1987)) is discussed. This dataset appears in Table \ref{tab1} (first and second columns). As it can be
seen, these data are heavily skewed to the right and overdispersed
since the sample variance, $s^2=0.12$, is greater than the sample mean, $\bar{x}=0.08$. Therefore, it is sensible to use an overdispersed (i.e. $\mathscr{NBL}$ distribution) discrete distribution to fit this dataset.

By taking as starting values the factorial moment estimates, the maximum likelihood method have been calculated for this dataset by using the $\mathscr{NBL}$ distribution. For the sake of comparison, others two--parameter discrete models, the negative binomial ($\mathscr{NB}$) and Poisson--inverse Gaussian ($\mathscr{PI}$) distributions have been used to describe this dataset. By using the
 chi--squared test to test the adherence to data of the aforementioned models with test statistic given by $\chi^2=\sum_{x} (p_x-\hat{p}_x)^2/\hat{p}_x$. In order to comply with the rule of five, the last three rows were combined. The $\mathscr{NBL}$ distribution provides the lowest value for the test statistics. By assuming that the theoretical distribution of the test statistics is $\chi^2_2$, the $p$-values are easily derived. Based on these $p$--values, there exists enough statistical evidence to not reject the null hypothesis that the data come from any of the models considered at the usual significance levels and therefore, there exists statistical evidence to not reject the data come none of the models. However, the test reject the $\mathscr{NB}$ and $\mathscr{PIG}$ distributions earlier than the $\mathscr{NBL}$  distribution. Additionally, by using the maximum of the likelihood function $\ell_{max}$ as criterion of comparison the $\mathscr{NBL}$ is preferable to $\mathscr{NB}$ and $\mathscr{PIG}$ distributions. The parameter estimates for $\mathscr{NBL}$ distribution, obtained by maximum likelihood estimation, are
$\hat r=0.486$ and $\hat \theta=6.381$ with standard errors given by 0.12 and 1.50, respectively. The estimated values by the other distributions can be viewed in Willmot (1987).
 The maximum likelihood estimates were also obtained by using the EM type algorithm introduced in this papers by using as the values mentioned above as starting values. In this case 155 iterations were needed to obtain the estimates $\hat\theta=6.663$ and $\hat r=0.509$ when the relative change of the log--likelihood function was smaller than $1\times 10^{-10}$ obtaining a value for $\ell_{max}=-1183.45$.

\begin{table}[htbp]
\caption{Observed and expected claim counts versus different models. See Willmot (1987).\label{tab1}}
\begin{center}
\begin{tabular}{crrrr}\hline
Counts & Observed & \multicolumn{3}{c}{Fitted}\\
\cline{3-5}
 & & $\mathscr{NB}$ & $\mathscr{PI}$ & $\mathscr{NBL}$ \\ \hline
 0 & 3719 & 3719.22 & 3718.58 & 3718.82 \\
 1 & 232  & 229.90 & 234.54 &  232.98    \\
 2 & 38  & 39.91 & 34.86 &  36.59    \\
 3 & 7  & 8.42 & 8.32 &  8.21      \\
 4 & 3  & 1.93 & 2.45 &  2.26      \\
 5 & 1  & 0.46 & 0.80 &  0.72      \\ \hline
 Total      & 4000 & 4000 & 4000 & 4000 \\ \hline \\

\multicolumn{2}{l}{$\chi^2_{2}$}  & 1.17 & 0.54 & 0.06\\
\multicolumn{2}{l}{$p$-value}  & 55.70\% & 76.20\% & 80.33\%\\
\multicolumn{2}{l}{$\ell_{\max}$} & --1183.550 & --1183.524 & --1183.430\\ \\ \hline
\end{tabular}
\end{center}
\end{table}

\section{Conclusion\label{s6}}

In this article an alternative representation of the Negative--Binomial--Lindley distribution has been proposed to explain positively skewed and overdispersed count data. The formulation of the model introduced in this work is more tractable the one presented in Zamani and Ismail (2010). Additionally, it includes some attractive properties such as the unimodality and overdispersion  A recurrence for the probabilities of the new distribution together with an integral equation
for the probability density function of the compound version, when
the claim severities are absolutely continuous, were presented. Finally, an EM type algorithm was also introduced when the triple mixture Poisson--Gamma--Lindley is considered to estimate the parameters of the model.

\section*{References}

\begin{description}

\item Balakrishnan, N. and Nevzorov, V.B. (2003). \textit{A Primer
on Statistical Distributions}. John Wiley, New York.

\item Ghitany, M., Al-Mutairi, D., Balakrishnan, N. and Al-Enezi,
L. (2013). Power Lindley distribution and associated inference.
\textit{Computational Statistics and Data Analysis}, 64, 20--33.

\item Ghitany, M., Atieh, B. and Nadarajah, S. (2008). Lindley
distribution and its applications. \textit{Mathematics and
Computers in Simulation}, 78, 4, 493--506.

\item G\'omez-D\'eniz, E. Sordo, M. and Calder\'in-Ojeda, E.
(2013). The log-Lindley distribution as an alternative to the Beta
regression model with applications in insurance. {\it Insurance:
Mathematics and Economics}, 54, 49--57.

\item Gradshteyn, I. and Ryzhik, I. (1994). \textit{Table of
Integrals, Series and Products 5th ed.} Jeffrey A., ed. Boston:
Academic Press.

\item Holgate. P. (1970). The modality of some compound Poisson
distribution. {\it Biometrika}, 56, 666--667.

\item Karlis, D. and Xekalaki, E. (2005). Mixed Poisson
distributions. {\it International Statistical Review}, 73, 35--59.

\item Klugman, S.A., Panjer, H.H. and Willmot, G.E. (2008).
\textit{Loss Models. From Data to Decisions}. Third Edition. John
Wiley, New Jersey.

\item Lindley, D. (1958). Fiducial distributions and Bayes's
theorem. {\it Journal of the Royal Statistical Society. Series B},
20, 1, 102--107.
\item Lord, D. and Geedipally, S. (2011). The negative binomial-Lindley distribution as a tool for analyzing crash data characterized by a large amount of zeros. {\it Accident Analysis and Prevention}, 43, 5, 1738--1742.
    
    \item Lynch, J. (1999). On conditions for mixtures of increasing failure rate distributions to have an increasing failure rate. {\it Probability in the Engineering and Informational Sciences}, 13, 1, 33--36.
  \item Rolski, t. Schmidli, H., Schmidt, V. and Teugel, J. (1999). \textit{Stochastic Processes for Insurance and Finance.} John Wiley \& Sons.   
    \item Sankaran, M. (1971). The discrete Poisson-Lindley distribution. {\it Biometrics}, 26, 1, 145--149. 
\item Sundt, B. and Vernic, R. (2009). \textit{Recursions for Convolutions and Compound Distributions with Insurance Applications.} Springer-Verlag, New York.  
\item Willmot, G. (1987). The Poisson-inverse Gaussian distribution as alternative to the negative binomial. {\it Scandinavian Actuarial Journal}, 113--137. 

\item Willmot, G. (1993). On recursive evaluation of mixed Poisson
probabilities and related quantities. {\it Scandinavian Actuarial
Journal}, 2, 114--133.

\item Wolfram, S. (2003). {\it The Mathematica Book}. Wolfram
Media, Inc.

\item Zamani, H. and Ismail, N. (2010). Negative binomial-Lindley
distribution and its applications. {\it Journal of Mathematics and
Statistics}, 6, 1, 4--9.

\end{description}

\end{document}